\documentclass[11pt]{article}

\usepackage{graphicx}
\usepackage{url}
\usepackage{amssymb}
\usepackage{amsmath}
\usepackage{amsthm}
\usepackage{ifthen}
\usepackage{verbatim}

\usepackage[left=1in,right=1in,top=1in,bottom=1in]{geometry}
\newcommand\hide[1]{} % Hide part of the text

\newtheorem{theorem}{Theorem}[section]

\newtheorem{lemma}[theorem]{Lemma}

\newcommand{\DEF}[1]{{\em #1\/}}
\newcommand{\NN}{\ensuremath{\mathbb N}}  % integer numbers

\begin{document}

\title{The Clique Problem in Ray Intersection Graphs\thanks{Research was supported in part 
		by the Slovenian Research Agency, program P1-0297 
		and projects J1-4106, L7-4119, BI-BE/11-12-F-002, and
		within the EUROCORES Programme EUROGIGA (projects GReGAS and ComPoSe) of the European Science Foundation. 
		Large part of the work was done while Sergio was visiting ULB.}}

\author{Sergio Cabello\thanks{Department of Mathematics, IMFM, and 
				Department of Mathematics, FMF, University of Ljubljana, Slovenia.	
				email: {\tt sergio.cabello@fmf.uni-lj.si}}
\and
        Jean Cardinal\thanks{Universit\'e Libre de Bruxelles (ULB), Brussels, Belgium. Email: {\tt jcardin@ulb.ac.be}}
\and
        Stefan Langerman\thanks{Ma\^itre de Recherches, Fonds de la Recherche Scientifique (F.R.S-FNRS), Universit\'e Libre de Bruxelles (ULB), Brussels, Belgium. Email: {\tt stefan.langerman@ulb.ac.be}}
}

\date{\today}

\maketitle

\begin{abstract}
Ray intersection graphs are intersection graphs of rays, or halflines, in the plane.
We show that any planar graph has an even subdivision whose complement is a ray intersection
graph. The construction can be done in polynomial time and implies that finding
a maximum clique in a segment intersection graph is NP-hard. This solves a 21-year old open problem posed by Kratochv\'il and Ne\v{s}et\v{r}il.
\end{abstract}

\thispagestyle{empty}

\newpage

\setcounter{page}{1}
%%%%%%%%%%%%%%%%%%%%%%%%%%%%%%%%%%%%%%%%%%%%%%%%%%%%%%%%%%%%%%%%%%%%%%%%%%%%%%%%%%%%%%%%%%%%%%%%%%%%%%%%%%%%%%%%%%%%%%
\section{Introduction}

The intersection graph of a collection of sets has one vertex for each set, and
an edge between two vertices whenever the corresponding sets intersect.
Of particular interest are families of intersection graphs corresponding to geometric sets in the plane.
In this contribution, we will focus on \DEF{segment intersection graphs}, intersection graphs of line segments in the plane.

In a seminal paper, Kratochv\'{\i}l and Ne\v{s}et\v{r}il~\cite{KN90} proposed to study the complexity of two classical
combinatorial optimization problems, the maximum independent set and the maximum clique, in geometric intersection graphs.
While those problems are known to be hard to approximate in general graphs (see for instance~\cite{EH00,H96}), their restriction
to geometric intersection graphs may be more tractable.
They proved that the maximum independent set problem remains NP-hard for segment intersection graphs, even if those
segments have only two distinct directions. It was also shown that in that case, the maximum clique problem
can be solved in polynomial time. The complexity of the maximum clique problem in general segment intersection 
graphs was left as an open problem, and remained so until now. In their survey paper 
``On six problems posed by {J}arik {N}e\v{s}et\v{r}il"~\cite{BRSSTW06}, Bang-Jensen {\em et al.}
describe this problem as being ``among the most tantalizing unsolved problem in the area".

Some progress has been made in the meanwhile. In 1992, Middendorf and Pfeiffer~\cite{MP92} showed, with a simple proof, 
that the maximum clique problem was NP-hard for intersection graphs of 1-intersecting curve segments that are either
line segments or curves made of two orthogonal line segments. They also give a polynomial time dynamic programming algorithm 
for the special case of line segments with endpoints of the form $(x,0),(y, i)$, with $i\in \{1,\dots k\}$ for some fixed $k$.
Another step was made by Amb\"uhl and Wagner~\cite{AW05} in 2005, who showed that the maximum clique problem was NP-hard
for intersection graphs of ellipses of fixed, arbitrary, aspect ratio. Unfortunately, this ratio must be bounded, which excludes
the case of segments.

\paragraph*{Our results.}

We prove that the maximum clique problem in segment intersection graphs is NP-hard. In fact, we prove the stronger
result that the problem is NP-hard even in \DEF{ray intersection graphs}, defined as intersection graphs of 
rays, or halflines, in the plane. This complexity result is a consequence of the following structural lemma: every planar graph has an 
even subdivision whose complement is a ray intersection graph. Furthermore, the corresponding set of rays has a natural polynomial 
size representation. Hence solving the maximum clique problem in this graph allows to recover the maximum independent set 
in the original planar graph, a task well known to be NP-hard~\cite{GJ77}. The construction is detailed in Section~\ref{sec:construction}.

\paragraph*{Related work.}
We prove that the complement of some subdivision of any planar graph can be represented as a segment intersection graph. 
Whether the complement of every planar graph is a segment intersection graph remains an open question. In 1998, Kratochv\'{\i}l and 
Kub\v{e}na~\cite{KK98} showed that the complement of any planar graph is the intersection graph of a set of convex polygons. More recently, 
Francis, Kratochv\'il, and Vysko\v{c}il~\cite{FKV10} proved that the complement of any partial 2-tree is a segment intersection graph.
Partial 2-trees are planar, and in particular every outerplanar graph is a partial 2-tree. The representability of planar graphs by 
segment intersection graphs, formerly known as Scheinerman's conjecture, was proved recently by Chalopin and Gon\c{c}alves~\cite{CG09}.

The maximum independent set problem in intersection graphs has been studied by Agarwal and Mustafa~\cite{AM04}. In particular, they proved that
it could be approximated within a factor $n^{1/2+o(1)}$ in polynomial time for segment intersection graphs. This has been recently improved 
by Fox and Pach~\cite{FP11}, who described, for any $\epsilon >0$, a $n^{\epsilon}$-approximation algorithm. In fact, their technique also applies to the maximum clique problem, and therefore $n^{\epsilon}$ is the best known approximation factor for this problem too.

In 1994, Kratochv\'{\i}l and Matou\v{s}ek~\cite{KM94} proved that the recognition problem for segment intersection 
graphs was in PSPACE, and was also NP-hard. It is still not clear whether it is NP-complete.

\paragraph*{Notation.}
For any natural number $m$ we use $[m]=\{1,\dots m\}$.
In a graph $G$, a \DEF{rotation system} is a list $\pi=( \pi_v )_{v\in V(G)}$,
where each $\pi_v$ fixes the clockwise order of the edges of $E(G)$ incident to $v$.
When $G$ is an embedded planar graph, the embedding uniquely defines a rotation system, 
which is often called a combinatorial embedding.
For the rest of the paper we use \DEF{ray} to refer to an \emph{open ray}, that is, 
a ray does not contain its origin.
Therefore, whenever two rays intersect they do so in the relative interior of both.
Since our construction does not use degeneracies, we are not imposing any restriction
by considering only open rays.
A subdivision of a graph $G$ is said to be \emph{even} if each edge of $G$ is subdivided an even
number of times.

%====================================================================================================================
\section{Construction}
\label{sec:construction}

Let us start providing an overview of the approach.
We first construct a set of curves that will form the \emph{reference frame}. 
This construction is quite generic and depends only on a parameter $k\in \NN$.
We then show that the complement of any tree has a special type of representation, 
called \emph{snooker representation}, which is constructed iteratively over the levels of the tree. 
The number of levels of the tree is closely related to $k$, 
the parameter used for the reference frame.
We then argue that if $G$ is a planar graph that consists of a tree $T$ and a few, special 
paths of length two and three, then the complement of $G$ can be
represented as an intersection graph of rays by extending a snooker representation of $T$.
Finally, we argue that any planar graph has an even subdivision that can be decomposed into a tree
and a set of paths of length two and three with the required properties. 

We first describe the construction using real coordinates. 
The construction does not rely on degeneracies, and thus we can slightly perturb the coordinates
used in the description.
This perturbation is enough to argue that a representation can be computed in polynomial time.
Then, using a relation between the independence number of a graph $G$ and an even subdivision of $G$,
we obtain that computing a maximum clique in a ray intersection graph is NP-hard. 

\subsection{Reference frame}

\begin{figure}
	\centering
	\includegraphics[page=2,width=.8\textwidth]{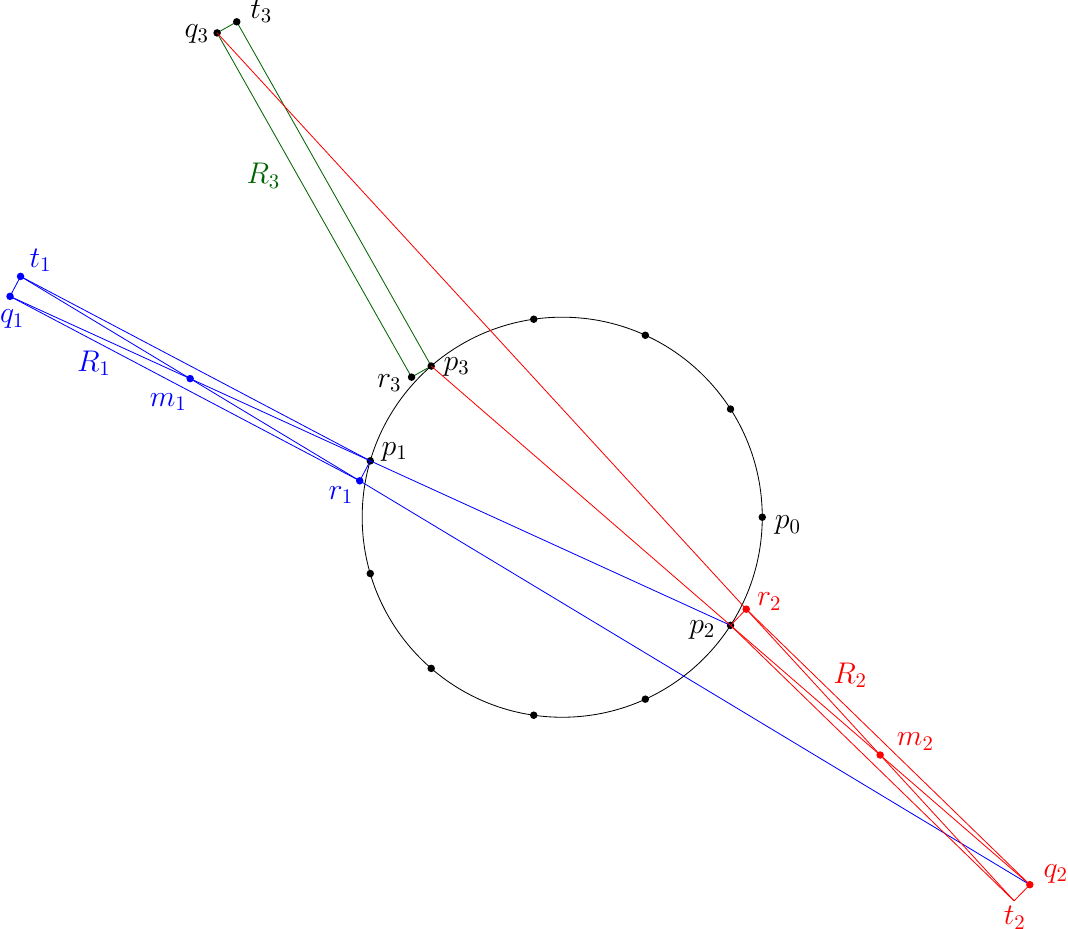}
	\caption{The points $p_0,\dots, p_k$ and the rectangles $R_1,\dots,R_{k-2}$.}
	\label{fi:points}
\end{figure}

Let $k$ be an \emph{odd} number to be chosen later. 
We set $\theta=\tfrac{k-1}{k}\pi$, and define for $i=0,\dots,k-1$ the points 
\[
	p_i= ( \cos (i\cdot \theta),\, \sin (i\cdot\theta)).
\]
The points $p_i$ lie on a unit circle centered at the origin; see Figure~\ref{fi:points}.
For each $i\in [k-2]$ we construct a rectangle $R_i$ as follows.
Let $q_i$ be the point $p_i+(p_i-p_{i+1})$, symmetric of $p_{i+1}$ with respect to $p_i$,
let $m_i$ be the midpoint between $p_i$ and $q_i$,
and let $t_i$ be, among the two points along the line through $q_{i+1}$ and $m_i$
with the property $|m_i t_i|=|m_i p_i|$, the one that is furthest from $q_{i+1}$.
We define $R_i$ to be the rectangle with vertices $p_i$, $t_i$, and $q_i$. The fourth
vertex of $R_i$ is implicit and is denoted by $r_i$. 
Any two rectangles $R_i$ and $R_j$ are congruent with respect to a rotation around the origin. 
We have constructed the rectangles $R_i$ in such way that, for any $i\in [k-2]$,
the line supporting the diagonal $p_i q_i$ of $R_i$ contains $p_{i+1}$
and the line supporting the diagonal $r_i t_i$ of $R_i$ contains $q_{i+1}$.

For each $i\in [k-2]$, let $\alpha_i$ be the arc of circle that is tangent
to both diagonals of $R_i$ and has endpoints $t_i$ and $r_i$ (see Figure~\ref{fi:alphas}). 
Note that the curves $\alpha_i$ and the rectangles $R_i$ have been chosen so that
any line that intersects $\alpha_i$ twice or is tangent to the curve $\alpha_i$ 
must intersect the curve $\alpha_{i+1}$. 
For any $i\in [k-2]$, 
let $\Gamma_i$ be the set of rays that intersect $\alpha_{i}$ twice or are tangent to $\alpha_{i}$
and have its origin on $\alpha_{i+1}$.
We also define $\Gamma_0$ as the set of rays with origin on $\alpha_1$ and passing through $p_0$. 
The rays of $\Gamma_i$ that are tangent to $\alpha_i$ will play a special role. 
In fact, we will only use rays of $\Gamma_i$ that are ``near-tangent" to $\alpha_i$.

\begin{lemma}
\label{lem:refint1}
When $|j-i|>1$, any ray from $\Gamma_i$ intersects any ray from $\Gamma_j$.
\end{lemma}

\begin{lemma}
\label{lem:refint2}
Any ray tangent to $\alpha_{i+1}$ at the point $x\in\alpha_{i+1}$
intersects any ray from $\Gamma_i$, except those having their origin at $x$.
\end{lemma}

Note that the whole construction depends only on the parameter $k$.
We will refer to it as \emph{reference frame}.

\begin{figure}
	\centering
	\includegraphics[page=3,width=.8\textwidth]{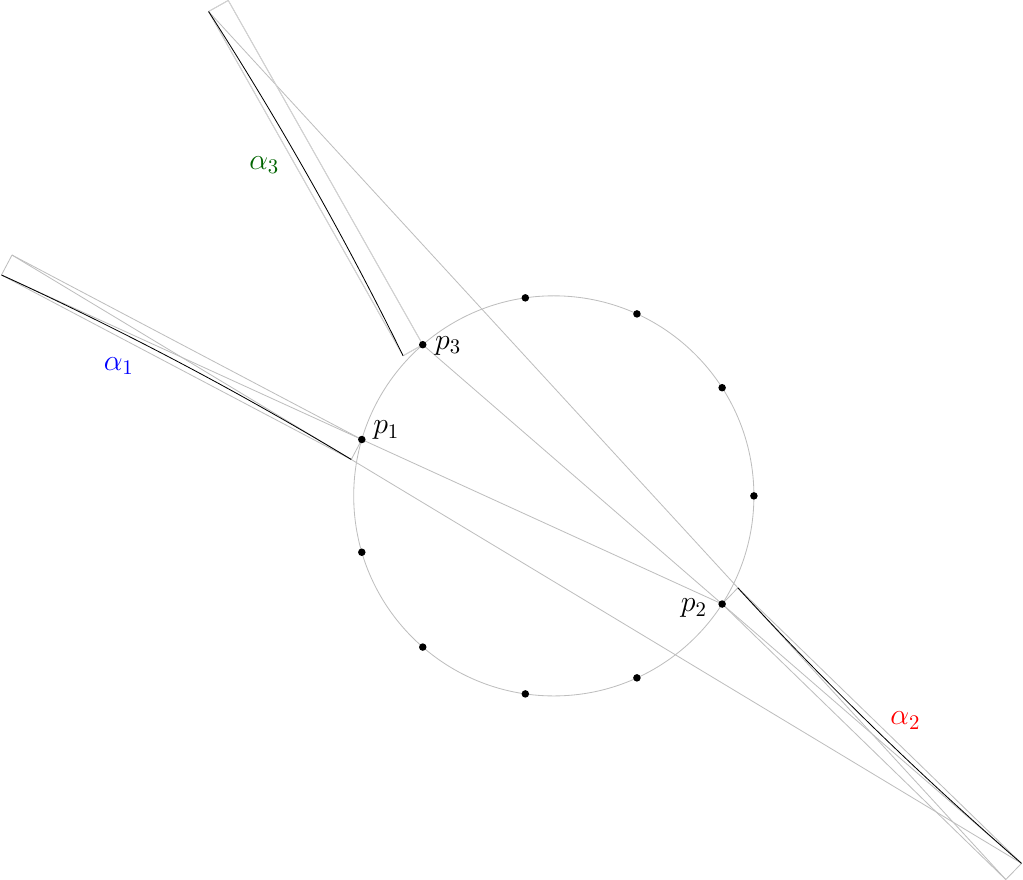}
	\caption{The circular arcs $\alpha_1,\alpha_2,\dots$.}
	\label{fi:alphas}
\end{figure}

\subsection{Complement of Trees}

Let $T$ be a graph with a rotation system $\pi_T$,
let $r$ vertex in $T$ and let $rs\in E(T)$ be an arbitrary edge incident to $r$.
The triple $(\pi_T,r,rs)$ induces a natural linear order $\tau=\tau(\pi_T,r,rs)$ 
on the vertices of $T$. This order $\tau$ corresponds to the order followed 
by a breadth-first traversal of $T$ from $r$ with the following additional restrictions: 
\begin{enumerate}
	\item[(i)] $s$ is the second vertex; 
	\item[(ii)] the children of any vertex $v$ are visited according to the clockwise order $\pi_v$;
	\item[(iii)] if $v\not= r$ has parent $v'$, the first child of $u$ of $v$ is such that
		$vu$ is the successor of $vv'$ in the clockwise order $\pi_v$.
\end{enumerate}
We say that vertices $v$ and $v'$ at the same level are \DEF{consecutive} when they are consecutive in $\tau$.
See Figure~\ref{fi:tree}. 
The linear order will be fixed through our discussion, so we will generally drop it from the notation.
Henceforth, whenever we talk about a tree $T$ and a linear order $\tau$ on $V(T)$, we assume
that $\tau$ is the natural linear order induced by a triple $(\pi_T,r,rs)$. In fact,
the triple $(\pi_T,r,rs)$ is implicit in $\tau$.
For any vertex $v$ we use $v^+$ for its successor and $v^-$ for its predecessor.

\begin{figure}
	\centering
	\includegraphics[page=4,width=.7\textwidth]{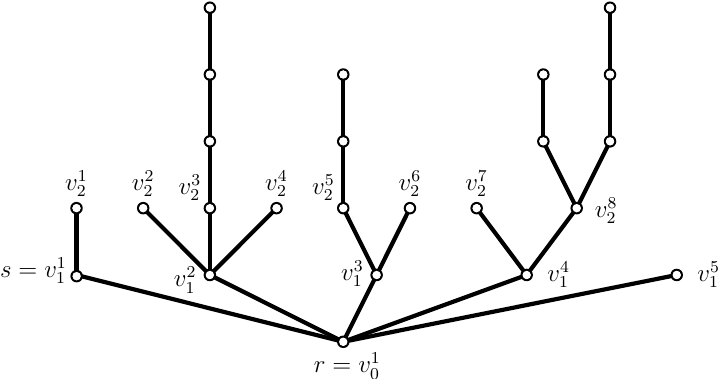}
	\caption{In this example, assuming that $\pi$ is as drawn, the linear order $\tau$ is 
		$r,s,v_1^1,\dots, v_{1}^5,v_2^1,\dots$.}
	\label{fi:tree}
\end{figure}

A \DEF{snooker representation} of the complement of an embedded tree $T$ 
with linear order $\tau$ is a representation of $T$ with rays that
satisfies the following properties:
\begin{itemize}
	\item[(a)] Each vertex $v$ at level $i$ in $T$ 
		is represented by a ray $\gamma_v$ from $\Gamma_i$. 
		Thus, the origin of $\gamma_v$, denoted by $a_v$, is on $\alpha_{i+1}$.
		Note that this imply that $k$ is larger than the depth of $T$. 
	\item[(b)] If a vertex $u$ has parent $v$, then $\gamma_u$ passes through the origin $a_v$ of $\gamma_v$.
		(Here it is relevant that we consider all rays to be open, as otherwise $\gamma_u$ and $\gamma_v$ would intersect.)
		In particular, all rays corresponding to the children of $v$ pass through the point $a_v$.
	\item[(c)] The origins of rays corresponding to consecutive vertices $u$ and $v$ of level $i$ 
		are consecutive along $\alpha_{i+1}$. 
		That is, no other ray in the representation
		has its origin on $\alpha_{i+1}$ and between the origins $\gamma_u$ and $\gamma_v$.
\end{itemize}

\begin{lemma}
\label{le:snooker}
	The complement of any embedded tree with a linear order $\tau$ has a snooker representation.
\end{lemma}
\begin{proof}
	Consider a reference frame with $k$ larger than the depth of $T$.
	The construction we provide is iterative over the levels of $T$.
	Note that, since we provide a snooker representation, it is enough to tell for each vertex $v\not= r$
	the origin $a_v$ of the ray $\gamma_v$. Property (b) of the snooker representation provides
	another point on the ray $\gamma_v$, and thus $\gamma_v$ is uniquely defined.
	The ray $\gamma_r$ for the root $r$ is the ray of $\Gamma_0$ with origin $a_r$
	in the center of $\alpha_1$. 
	
	Consider any level $i>1$ and assume that we have a representation of the vertices
	at level $i-1$. Consider a vertex $v$ at level $i-1$ and let 
	$u_1,\dots, u_d$ denote its $d$ children. If the successor $v^+$ of $v$ is also at level $i-1$,
	we take $a_v^+= a_{v^+}$, and else we take $a_v^+$ to be an endpoint of $\alpha_i$ such
	that no other origin is between the endpoint and $a_v$.
	See Figure~\ref{fi:snooker}.
	Similarly, if the predecessor $v^-$ of $v$ is at level $i-1$,
	we take $a_v^-= a_{v^-}$, and else we take $a_v^-$ to be an endpoint of $\alpha_i$ such
	that no other origin is between the endpoint and $a_v$.
	(If $v$ is the only one vertex at level $i$, we also make sure that $a_v^-\not= a_v^+$.) 
	Let $\ell_v^+$ be the line through $a_v$ and $a_v^+$.
	Similarly, let $\ell_v^-$ be the line through $a_v$ and $a_v^-$.
	We then choose the points $a_{u_1},\dots, a_{u_d}$ 
	on the portion of $\alpha_{i+1}$ contained between $\ell_v^+$ and $\ell_v^-$ such that
	the $d+2$ points $\ell_v^-\cap \alpha_{i+1}, a_{u_1},\dots, a_{u_d}, \ell_v^+\cap \alpha_{i+1}$ 
	are regularly spaced.
	Since the ray $\gamma_{u_j}$ has origin $a_{u_j}$ and passes through $a_v$, this finishes the description
	of the procedure. 
	Because $a_{u_j}$ lies between $\ell_v^+$ and $\ell_v^-$, 
	the ray $\gamma_{u_j}$ either intersects $\alpha_i$ twice or is tangent to $\alpha_i$,
	and thus $\gamma_{u_j}\in \Gamma_{i}$.
	
	\begin{figure}
	\centering
	\includegraphics[page=5,width=.9\textwidth]{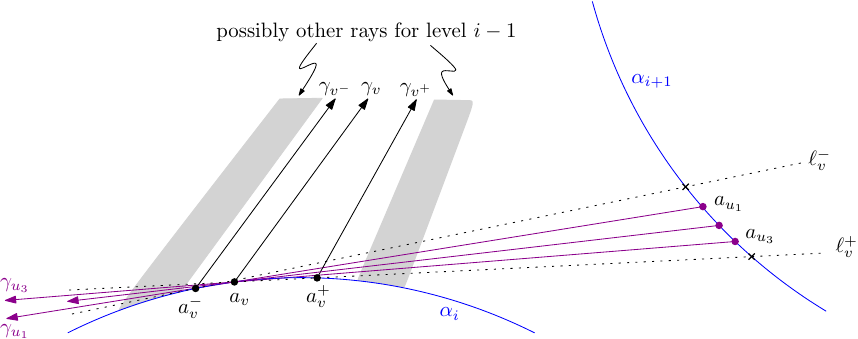}
	\caption{An example showing the construction of a snooker representation when $v$ has three children
		and $v^+$ and $v^-$ are at the same level as $v$.}
	\label{fi:snooker}
	\end{figure}

	Recall that any ray from $\Gamma_i$ intersects any ray from $\Gamma_j$ when $|j-i|>1$.
	Therefore, vertices from levels $i$ and $j$, where $|i-j|>1$, intersect.
	For vertices $u$ and $v$ at levels $i-1$ and $i$, respectively, 
	the convexity of the curve $\alpha_{i}$ and the choices for $a_v$ imply that $\gamma_v$ 
	intersects $\gamma_u$ if and only if $u$ is not the parent of $v$ in $T$. 
	For vertices $u$ and $v$ at the same level $i$, the rays $\gamma_u$ and $\gamma_v$
	intersect: if they have the same parent $w$, then they intersect on $a_w$,
	and if they have different parents, the order of their origins $a_u$ and $a_v$ on $\alpha_{i+1}$ 
	and the order of their intersections with $\alpha_{i}$ are reversed.
\end{proof}

\subsection{A tree with a few short paths}
\label{sec:construction}
Let $T$ be an embedded tree with a linear order $\tau$.
An \DEF{admissible extension} of $T$ is a graph $P$ with the following properties
\begin{itemize}
	\item $P$ is the union of vertex-disjoint paths (i.e., two
          paths don't share internal vertices but they are allowed to
          share endpoints);
	\item each maximal path in $P$ has 3 or 4 vertices;
	\item the endpoints of each maximal path in $P$ are leaves of $T$
		that are consecutive and at the same level;
	\item the internal vertices of any path in $P$ are not vertices of $V(T)$.
\end{itemize}
Note that $T+P$ is a planar graph because we only add paths between consecutive leaves.

\begin{lemma}
\label{le:adjacency}
	Let $T$ be an embedded tree and let $P$ be an admissible extension of $T$.
	The complement of $T+P$ is a ray intersection graph.
\end{lemma}
\begin{proof}
	We construct a snooker representation of $T$ using Lemma~\ref{le:snooker} where $k$, the number
	of levels, is the depth of $T$ plus 2 or 3, whichever is odd.
	We will use a local argument to represent each maximal path of $P$,
	and each maximal path of $P$ is treated independently.
	It will be obvious from the construction that rays corresponding to vertices
	in different paths intersect. 
	We distinguish the case where the maximal path has one internal vertex or two.
	
	Consider first the case of a maximal path in $P$ with one internal vertex.
	Thus, the path is $uwv$ where $u$ and $v$ are consecutive leaves in $T$ and $w\notin V(T)$
	is not yet represented by a ray.
	The origins $a_u$ and $a_v$ of the rays $\gamma_u$ and $\gamma_v$, respectively, 
	are distinct and consecutive along $\alpha_{i+1}$ because we have a snooker representation.
	We thus have the situation depicted in Figure~\ref{fi:scenario1}. 
	We can then just take the $\gamma_w$ to be the line through $a_u$ and $a_v$.
         (This line can also be a ray with an origin sufficiently far away.)
	This line intersects the ray of any other vertex, different that $\gamma_u$ and $\gamma_v$.
	
	\begin{figure}
		\centering
		\includegraphics[page=6,width=.9\textwidth]{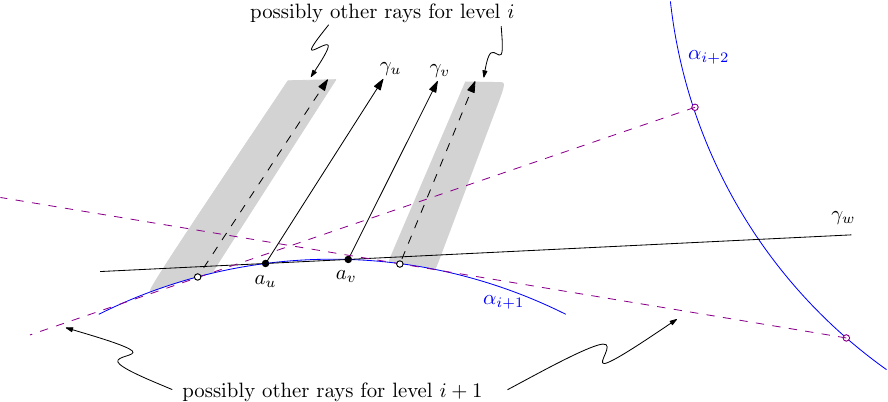}
		\caption{Case 1 in the proof of Lemma~\ref{le:adjacency}: 
			adding a path with one internal vertex.}
		\label{fi:scenario1}
	\end{figure}
	
	Consider now the case of a maximal path in $P$ with two internal vertices.
	Thus, the path is $uww'v$ where $u$ and $v$ are consecutive leaves in $T$ and $w,w'\notin V(T)$
	In this case, the construction depends on the relative
	position of the origins $a_u$ and $a_v$, and we distinguish two scenarios: 
	(i) shifting the origin $a_u$ of ray $\gamma_u$ towards $a_v$ while maintaining the slope 
	introduces an intersection between $\gamma_u$ and the ray for the parent of $u$
	or (ii) shifting the origin $a_v$ of ray $\gamma_v$ towards $a_u$ while maintaining the slope
	introduces an intersection between $\gamma_v$ and the ray for the parent of $v$.
	Note that exactly one of the two options must occur.
	
	Let us consider only scenario (i), since the other one is symmetric; 
	see Figure~\ref{fi:scenario2}.
	We choose a point $b_w$ on $\alpha_{i+1}$ between $a_u$ and $a_v$ very near $a_u$
	and represent $w$ with a ray $\gamma_w$ parallel to $\gamma_u$ with origin $b_w$.
	Thus $\gamma_w$ does not intersect $\gamma_u$ but intersects any other ray because we are in scenario (i). 
	Finally, we represent $w'$ with the line $\gamma_{w'}$ through points $b_w$ and $a_v$.
	With this, $\gamma_{w'}$ intersects the interior of any other ray but $\gamma_w$ and $\gamma_v$,
	as desired. Note that $\gamma_{w'}$ also intersects the rays for vertices in the same level because
	those rays are near-tangent to $\alpha_{i+2}$, which is intersected by $\gamma_w$.

	Note that in every case, since the rays $\gamma_w$ and $\gamma_{w'}$, respectively, are actually lines, 
         no new ray has its origin on $\alpha_{i+2}$. Hence rays having consecutive origins on $\alpha_{i+2}$ 
         remain so after the inclusion of a path in level $i+1$.
\end{proof}
	\begin{figure}
		\centering
		\includegraphics[page=7,width=.9\textwidth]{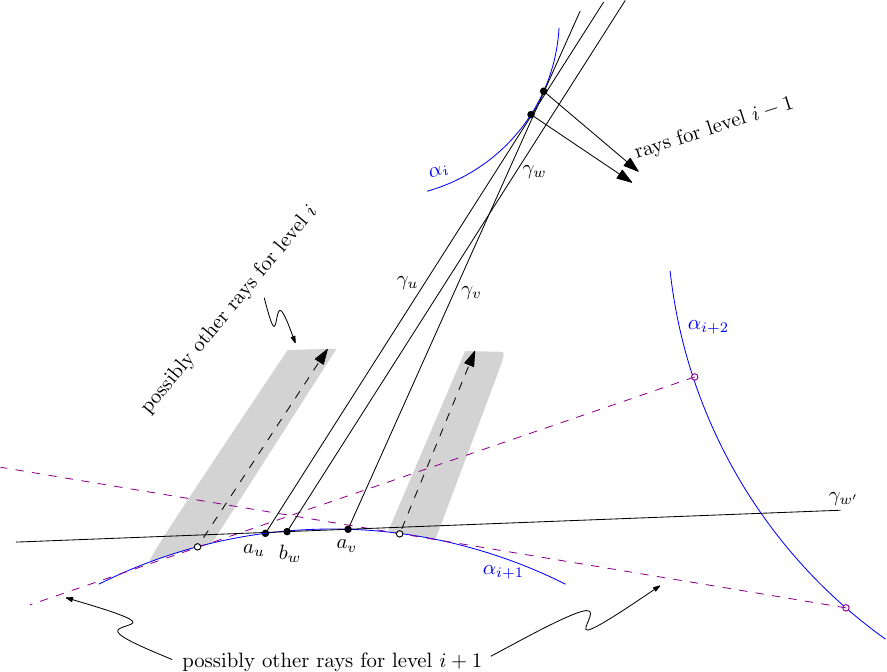}
		\caption{Case 2(i) in the proof of Lemma~\ref{le:adjacency}: 
			adding a path with two internal vertices.}
		\label{fi:scenario2}
	\end{figure}

\begin{lemma}
\label{le:findT+P}
	Any embedded planar graph $G$ has an even subdivision $T+P$, where $T$ is an embedded tree and $P$
	is an admissible extension of $T$.
	Furthermore, such $T$ and $P$ can be computed in polynomial time.
\end{lemma}
\begin{proof}
	Let $r$ be an arbitrary vertex in the outer face $f$ of $G$.
	Let $B$ be the set of edges in a BFS tree of $G$ from $r$.
	With a slight abuse of notation, we also use $B$ for the BFS tree itself.
	Let $C=E(G)\setminus B$. In the graph $G^*$ dual to $G$, 
	the edges $C^*=\{c^*\mid c\in C\}$ are a spanning tree of $G^*$, 
	which with a slight abuse of notation we also denote by $C^*$. 
	We root $C^*$ at the node $f^*$, corresponding to the outer face. This is illustrated on Figure~\ref{fi:findT+P}.	
	\begin{figure}
		\centering
		\includegraphics[page=8,width=.9\textwidth]{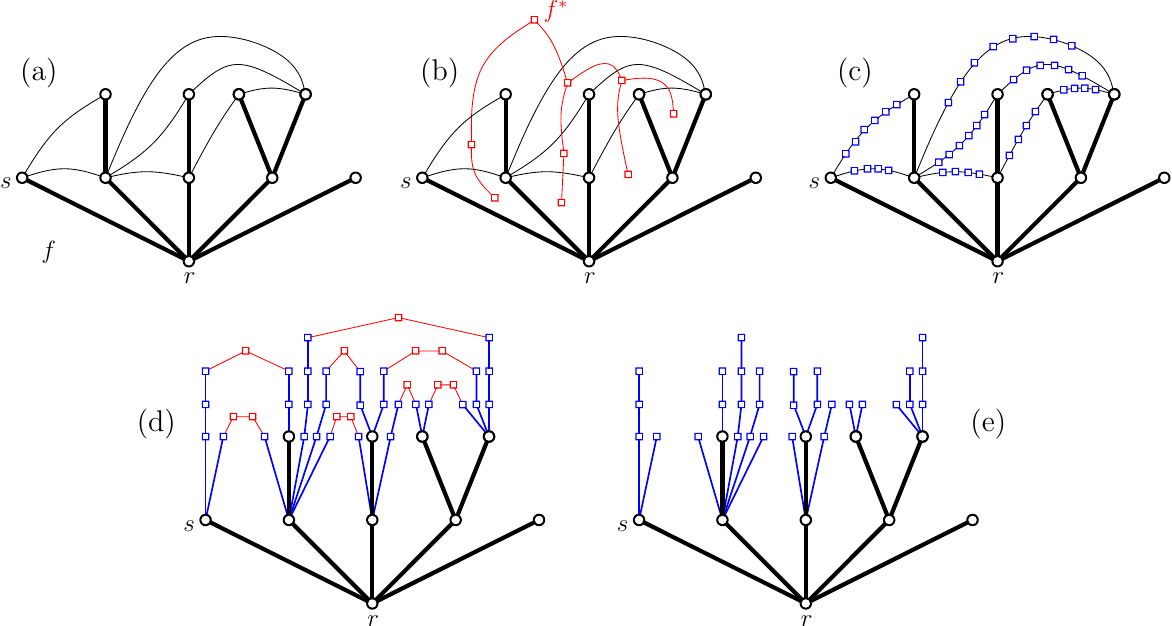}
		\caption{Figure for the proof of Lemma~\ref{le:findT+P}.
			(a) A planar graph with a BFS from $r$ in bold.
			(b) The spanning tree $C^*$ of the dual graph.
			(c) The subdivision of the edges from $C$.
			(d) The resulting subdivided graph with $P$ marked in red and thinner lines.
			(e) The resulting graph $T$ after the removal of $P$, drawn such that
				the height of the vertices corresponds to their level.
			}
		\label{fi:findT+P}
	\end{figure}
	
	We define for each edge $e\in C$ the number $k_e$ of subdivisions it will undertake
	using a bottom-up approach.
	Any edge $e\in C$ that is incident to a leaf of $C^*$ gets assigned $k_e=4$.
	For any other edge $e\in C$, we define $k_e$ as $2$ plus the maximum $k_{e'}$
	over all descendants $(e')^*\in  C^*$ of $e^*$ in $C^*$.
	Let $H$ be the resulting subdivision. 
	For any edge $e\in C$, let $Q_e$ be the path in $H$ that corresponds to the subdivision of $e$.
	We use in $H$ the combinatorial embedding induced by $G$.
	
	We can now compute the tree $T$ and the paths $P$. Since $B$ is a BFS tree, every edge $e\in C$ connects two vertices that are either at the same level in $B$, or at two successive levels.
	For every edge $e\in C$ that connects two vertices
	at the same level in $B$, let $P_e$ be the length-two subpath 
	in the middle of $Q_e$.
	For every edge $uv\in C$ that connects vertex $u$ at level $i$ 
	to vertex $v$ at level $i+1$ in $B$, 
	let $P_e$ be the length-three subpath obtained from the length-two subpath
	in the middle of $Q_e-u$. 
	We then take $P=\bigcup_{e\in C} P_e$ and
	take $T$ to be the graph $H-P$, after removing isolated vertices. 
	In $T$ we use the rotation system inherited from $H$ and use the edge $rs$ to define the linear order,
	where $rs$ is an edge in $f$.
	
	It is clear from the construction that $T+P=H$ is an even subdivision of $G$.
	We have to check that $P$ is indeed an admissible extension of $T$.
	The maximal paths of $P$ are vertex disjoint and connect leaves of $P$
	because the paths $Q_e$, $e\in C$, are edge-disjoint and each $P_e$ is strictly
	contained in the interior of $Q_e$.
	Since in $H-P$ we removed isolated vertices, it is clear that 
	no internal vertex of a path of $P$ is in $T$. The graph $T$ is indeed a tree
	because, for every edge $e\in C$, we have removed some edges from its subdivision $Q_e$.
	Since $B$ is a BFS tree and $P_{uv}$ is centered within $Q_{uv}$, when $u$ and $v$ are at the same level, or within $Q_{uv}-u$, when $u$ is one level below $v$, 
	the maximal paths of $P$ connect vertices that are equidistant from $r$ in $H$.
	
	It remains to show that the endpoints of any maximal path in $P$ are consecutive vertices
	in $T$. This is so because of the inductive definition of $k_e$. The base case is when $e\in C$ is incident to a leaf of $C^*$, and is subdivided $k_e=4$ times. The edge $e$ either connects two vertices at the same level in $B$, or at two successive levels. In both cases it can be checked that $P_e$ connects two consecutive vertices in $T$. The inductive case is as follows. We let $k_e$ be equal to $2$ plus the maximum $k_{e'}$ over all descendants $(e')^*\in  C^*$ of $e^*$ in $C^*$. By induction, all the corresponding $P_{e'}$ connect two consecutive vertices of $T$, say at level $i$ in $T$. By definition, $P_e$ will connect two consecutive vertices of $T$ at level $i+1$. 
		
	The construction of $T$ and $P$ only involves computing the BFS tree $B$, the spanning tree $C^*$, and the values $k_e$, which can clearly be done in polynomial time.
\end{proof}

Combining lemmas~\ref{le:adjacency} and \ref{le:findT+P} directly yields the following.
\begin{theorem}
	Any planar graph has an even subdivision whose complement is a ray intersection graph. Furthermore, this subdivision can be computed in polynomial time.
\end{theorem}

%========================================================================
\subsection{Polynomial-time construction}
\label{sec:NP}

The construction of the ray intersection graph in Lemmas~\ref{le:snooker} and \ref{le:adjacency}, uses real coordinates. We wish to prove that the maximum clique problem is NP-hard even when a geometric description of the ray intersection graph is given as input. Hence we need to argue how to carry out this construction using integer coordinates, each using a polynomial number of bits. In what follows, we let $n=|V(T+P)|$.

\begin{lemma}
	In the construction of Lemma~\ref{le:adjacency},
	\begin{itemize}
		\item any two points are at distance at least $n^{-O(n)}$ and at most $O(n)$;
		\item the distance between any line through two origins and any other point
			is at least $n^{-O(n)}$, unless the three points are collinear.
	\end{itemize}
\end{lemma}
\begin{proof}
	Recall that the circle containing points $p_0,p_1,\dots$ has radius $1$.
	The rectangles $R_i$ are all congruent and have two diagonals of length $|p_0 p_1|=\Theta(1)$.
	The small side of rectangle $R_i$ has size $\Theta(1/n)$ and
	both diagonals of $R_i$ form an angle of $\Theta(1/n)$. It follows that
	the center of the circles supporting $\alpha_i$ have coordinates $\Theta(n)$.
	For points from different curves $\alpha_i$ there is at least a separation of $\Theta(1/n)$.

	We first bound the distance between the origins for the rays representing vertices of $T$.
	Let us refer to the origins of rays lying on $\alpha_i$
	and the extremes of $\alpha_i$ as \emph{features} on the curve $\alpha_i$ . 
	Let $\delta_i$ be the minimum separation between any two features on the curve $\alpha_i$.
	On the curve $\alpha_1$ there are three features: the two extremes of $\alpha_1$
	and the origin $a_r$ of $\gamma_r$, which is in the middle of $\alpha_1$.
	Since $\alpha_1$ has length $\Omega(1)$, it follows that $\delta_1=\Omega(1)$.
	
	We will bound the ratio $\delta_{i+1}/\delta_{i}$ for $i\ge 1$. 
	Consider the construction of Lemma~\ref{le:snooker} to determine the features
	on $\alpha_{i+1}$.
	By induction, any two consecutive features along $\alpha_{i}$
	are separated at least by $\delta_i$. Since $\alpha_{i}$ is supported
	by a circle of radius $\Theta(n)$, 
	the lines $\ell_v^+$ and $\ell_v^-$ form an angle of 
	at least $\Omega(\delta_i/n)$.
	This implies that the points
	$\ell_v^-\cap \alpha_{i+1}, a_{u_1},\dots, a_{u_d}, \ell_v^+\cap \alpha_{i+1}$ 
	have a separation of $\Omega(\delta_i/(nd))$. It follows that 
	$\delta_{i+1}=\Omega(\delta_i/(nd))=\Omega(\delta_i/n^2)$, and thus $\delta_{i+1}/\delta_i=\Omega(1/n^2)$. 
	Since $T$ has depth at most $n$, all features are at distance at least $n^{-O(n)}$.
		
	We can now argue that the origins of the rays used in the construction of Lemma~\ref{le:adjacency}
	also have a separation of $n^{-O(n)}$.
	In Case 1, we just add a line through previous features.
	In Case 2, it is enough to place $b_w$ at distance $|a_u a_v|/n$ from $a_u$, and thus the features
	keep a separation of at least $n^{-O(n)}$.
	
	The second item is a consequence of the fact that
	if a point $(a,b)$ is not on the line through $(x,y)$ and $(x',y')$
	then its distance is at least $|y+ \tfrac{y'-y}{x'-x}(a-x) - b|$.
\end{proof}

We can now give the following algorithmic version of Lemma~\ref{le:adjacency}.

\begin{lemma}
\label{le:adjacency2}
	Let $T$ be an embedded tree and let $P$ be an admissible extension of $T$.
	We can find in polynomial time a family of rays described with integer coordinates
	whose intersection graph is isomorphic to the complement of $T+P$.
\end{lemma}
\begin{proof}
        Recall that the construction in Lemma~\ref{le:adjacency} consists of first constructing 
        a snooker representation of $T$, as described in Lemma~\ref{le:snooker}, then adding the
        rays corresponding to the paths in the extension.
        In the first part, each new point is created by: (a) computing the two intersections
        between two lines $\ell^-_v$ and $\ell^+_v$ through two existing points and a curve
        $\alpha_{i+1}$, then by (b) equally spacing $O(n)$ points between those two
        points on $\alpha_{i+1}$ (see Figure~\ref{fi:snooker}).
        In the second part, the only new points of the construction are the points $b_w$ 
        added on $\alpha_{i+1}$, between $a_u$ and $a_v$ (see Figure~\ref{fi:scenario2}). 
        We refer to this latter case as (c).

        Consider that each point is moved by a distance $< \varepsilon$
        after it is constructed. This may cause any point further
        constructed from those to move as well. In cases (b) and (c), the new
        points would be moved by no more than $\varepsilon$ as well.  

        In case (a) however, the error could be amplified. Let $a$ and
        $b$ be two previously constructed points, and suppose they
        are moved to $a'$ and $b'$, within a radius of
        $\varepsilon$. By the previous lemma, the distance between $a$
        and $b$ is at least $n^{-O(n)}$ and so the angle between the line $ab$
        and $a'b'$ is at most $\varepsilon n^{O(n)}$. Because the
        radius of the supporting circle of $\alpha_i$ is $\Theta(n)$,
        the distance between $ab \cap \alpha_i$
        and $a'b' \cap \alpha_i$ is at most $O(n)\varepsilon n^{O(n)}
        = \varepsilon n^{O(n)}$.
        
        Therefore, an error in one construction step expands by a
        factor $n^{O(n)}$. Now observe that each point of type
        (a) is constructed on the next level, and a point
        of type (b) is always constructed from points of type
        (a). Therefore, as there are $O(n)$ levels, an error is
        propagated at most $O(n)$ times, and the total
        error propagated from moving each constructed point by a
        distance $< \varepsilon$ is at most $\varepsilon n^{O(n^2)}$.

	By the previous lemma, there is a constant $c$ such that any origin
	in the construction of Lemma ~\ref{le:adjacency} is separated
        at least by a distance $A=1/n^{cn}$ from any other origin and
        any ray not incident to it. Therefore, by choosing
        $\varepsilon = n^{-c'n^3}$ for $c'$ large enough, the total
        propagation will never exceed $A$, and therefore perturbing
        the basic construction points of the reference frame and each
        further constructed point by a distance $<\varepsilon$ will
        not change the intersection graph of the modified rays.

        Therefore, to construct the required set of rays in polynomial
        time, we multiply every coordinate by the smallest power of
        $2$ larger than $1/\varepsilon$ and snap every constructed
        point to the nearest integer while following the construction
        of Lemma~\ref{le:adjacency}. Each coordinate can then be
        represented by $O(n^3)$ bits.
\end{proof}

Let $\alpha(G)$ be the size of the largest independent set in a graph $G$.
The following simple lemma can be deduced from the observation that subdividing
an edge twice increases the independence number by exactly one.

\begin{lemma}\label{lem:even-sub}
	If $G'$ is an even subdivision of $G$ where each edge $e\in E(G)$ is subdivided $2k_e$ times,
	then $\alpha(G')=\alpha(G)+\sum_e k_e$
\end{lemma}

By combining Lemmas~\ref{le:findT+P}, \ref{le:adjacency2}, and \ref{lem:even-sub}, we obtain:

\begin{theorem}
	Finding a maximum clique in a ray intersection graph is NP-hard, even when the input
	is given by a geometric representation as a set of rays.
\end{theorem}
 
\bibliographystyle{plain}
\bibliography{paper}
\end{document}